\documentclass[conference]{IEEEtran}
\IEEEoverridecommandlockouts

\usepackage[utf8]{inputenc}

\usepackage{amsmath, amsthm, amssymb, amsfonts} 
\usepackage{nicefrac}
\usepackage{mathtools}
\usepackage{bm, bbm}
\usepackage[scr=boondoxo,scrscaled=1.05]{mathalfa}

\usepackage[square,numbers]{natbib}

\usepackage{enumitem}

\usepackage{xcolor}
\usepackage{comment}

\usepackage{blindtext}

\definecolor{orange}{rgb}{1,0.4,0.0}

\DeclarePairedDelimiterXPP{\KL}[2]{D_\textnormal{KL}}{(}{)}{}{%
#1\:\delimsize\|\:#2%
}

\DeclarePairedDelimiterXPP{\RD}[2]{D_{\alpha}}{(}{)}{}{%
#1\:\delimsize\|\:#2%
}

\DeclarePairedDelimiterXPP\Prob[1]{\mathbb{P}}{\lbrace}{\rbrace}{}{

#1}

\DeclarePairedDelimiterXPP{\lnorm}[2]{}{\lVert}{\rVert}{_{#2}}{#1}

\newcommand{\bA}{\ensuremath{\mathbb{A}}}
\newcommand{\bE}{\ensuremath{\mathbb{E}}}
\newcommand{\bR}{\ensuremath{\mathbb{R}}}

\newcommand{\bP}{\ensuremath{\mathbb{P}}}

\newcommand{\cW}{\ensuremath{\mathcal{W}}}
\newcommand{\cX}{\ensuremath{\mathcal{X}}}

\newcommand{\cZ}{\ensuremath{\mathcal{Z}}}
\newcommand{\cN}{\ensuremath{\mathcal{N}}}

\newcommand{\stack}[2]{\stackrel{\mathclap{(#1)}}{#2}}

\newtheoremstyle{mytheoremstyle} 
    {\topsep}                    
    {\topsep}                    
    {\itshape}                   
    {}                           
    {\bf}                        
    {.}                          
    {.5em}                       
    {}  

\theoremstyle{mytheoremstyle}
\newtheorem{lemma}{Lemma}

\newtheorem{proposition}{Proposition}

\newtheorem{corollary}{Corollary}
\newtheorem{remark}{Remark}

\newcounter{MYtempeqncnt}

\title{On Random Subset Generalization Error Bounds and the Stochastic Gradient Langevin Dynamics Algorithm}
\author{%
  Borja Rodr\'iguez-G\'alvez, Germ\'an Bassi, Ragnar Thobaben, and Mikael Skoglund %
  \thanks{This work was supported in part by the Knut and Alice Wallenberg Foundation, the Swedish Foundation for Strategic Research, and the Swedish Research Council.
  } \\
  Division of Information Science and Engineering (ISE)\\
  KTH Royal Institute of Technology\\
  \texttt{\{borjarg, germanb, ragnart, skoglund\}@kth.se}
}
\date{August 2020}

\begin{document}

\maketitle

\begin{abstract}
In this work, we unify several expected generalization error bounds based on random subsets using the framework developed by~\citet{hellstrm2020generalization}.
First, we recover the bounds based on the individual sample mutual information from~\citet{bu2020tightening} and on a random subset of the dataset from~\citet{negrea2019information}.
Then, we introduce their new, analogous bounds in the randomized subsample setting from~\citet{steinke2020reasoning}, and we identify some limitations of the framework.
Finally, we extend the bounds from~\citet{haghifam2020sharpened} for Langevin dynamics to stochastic gradient Langevin dynamics and we refine them for loss functions with potentially large gradient norms.
\end{abstract}

\section{Introduction}

A learning algorithm $\bA$ is a mechanism that takes as an input a sequence $S = (Z_1, \ldots, Z_N)$ of $N$ i.i.d. samples $Z_i \in \cZ$ from $P_Z$, or a \emph{dataset}, and outputs a \emph{hypothesis} $W \in \cW$ by means of the conditional probability distribution $P_{W|S}$. 


We measure how well a hypothesis $W$ describes a sample $Z$ using a loss function $\ell: \cW \times \cZ \rightarrow \bR^+$.
Hence, a hypothesis describes the samples from a population $P_Z$ well when its \emph{population risk}, i.e., 
$
	L_{P_Z}(W) \triangleq \bE_{P_Z}[\ell(W,Z)],
$
is low.
However, $P_Z$ is often not available and we consider instead the \emph{empirical risk} on the dataset $S$, i.e.,
$
	L_{S}(W) \triangleq \frac{1}{N} \sum_{i=1}^N \ell(W,Z_i),
$
as a proxy.
Therefore, it is of interest to study the discrepancy between the population and empirical risks, which is defined as the \emph{generalization error}:
\begin{equation*}
	\textnormal{gen}(W,S) \triangleq  L_{P_Z}(W) - L_S(W).
\end{equation*}
By characterizing the generalization error of learning algorithms, we can intuit how far the empirical risk is from the real population risk.
Classical approaches have bounded the generalization error in expectation and in probability (PAC Bayes) either by measuring the complexity of the hypothesis space $\cW$ or by exploring properties of the learning algorithm itself (see~\citep{shalev2014understanding} for an overview of traditional approaches).

More recently, \citet{xu2017information}, based on~\citep{russo2019much}, found that the {expected generalization error} $\bE_{P_{W,S}} \big[ \textnormal{gen}(W,S) \big]$ is bounded from above by a function that depends on the mutual information between the hypothesis $W$ and the dataset $S$ with which it is trained, i.e.,  $I(W;S)$.
Similarly, \citet{bu2020tightening} and \citet{negrea2019information} found that it is also bounded from above by a function that depends on the dependency between the hypothesis and an individual sample, $I(W;Z_i)$, and on the hypothesis and a subset $S_J$ of the dataset $\KL[\big]{P_{W|S}}{P_{\smash{W|S_{J^c}}}}$.

After that, \citet{steinke2020reasoning} introduced a more structured setting.
They consider a super-sample of $2N$ i.i.d. samples $\tilde{Z}_i$ from $P_Z$, i.e., $\tilde{S} = (\tilde{Z}_1, \ldots, \tilde{Z}_{2N})$.
This super-sample is then used to construct the dataset $S$ by choosing between the samples $\tilde{Z}_i$ and $\tilde{Z}_{i+N}$ using a Bernoulli random variable $U_i$ with probability $1/2$; i.e., $Z_i = \tilde{Z}_{i + U_i N}$.
In this paper, we will distinguish between these two settings as the \emph{standard setting} and the \emph{randomized subsample setting}.\footnote{Note that in~\citep{hellstrm2020generalization}, the latter is called the random-subset setting. However, this may cause confusion with the random subset bounds in the present work.}

In the randomized subsample setting, we can define the \emph{empirical generalization error} as the difference between the empirical risk on the samples from $\tilde{S}$ not used to obtain the hypothesis $W$ from the algorithm, i.e., $\bar{S} = \tilde{S} \setminus S$, and the empirical risk on the dataset $S$; i.e., \[
    \widehat{\textnormal{gen}}(W,\tilde{S},U) \triangleq 
    \frac{1}{N} \! \sum_{i=1}^N \! \Big( \ell \big( W, \tilde{Z}_{i+(1-U_i)N} \big) \!- \ell \big( W,\tilde{Z}_{i+U_iN} \big) \! \Big).
\]
Then, we may note that the expected value of the empirical and the (standard) generalization errors coincide; i.e., $\bE_{P_{\smash{W,\tilde{S},U}}} \big[ \widehat{\textnormal{gen}}(W,\tilde{S},U) \big] = \bE_{P_{W,S}} \big[ \textnormal{gen}(W,S) \big]$, where $U$ is the sequence of $N$ i.i.d. Bernoulli trials $U_i$.
Moreover, it was also shown that the expected generalization error is bounded from above by a function of the conditional mutual information between the hypothesis $W$ and the Bernoulli trials $U$, given the super-sample $\tilde{S}$~\citep{steinke2020reasoning}, i.e., $I(W;U|\tilde{S})$, and by a subset $U_J$ of the Bernoulli trials~\citep{haghifam2020sharpened}, i.e., $\KL[\big]{P_{\smash{W|U,\tilde{S}}}}{P_{\smash{W|U_{J^c},\tilde{S}}}}$.

A first step towards unifying these results was given by
\citet{hellstrm2020generalization}, who introduced a framework that allowed them to recover the mutual information $I(W;S)$ and conditional mutual information $I(W;U|\tilde{S})$ expected generalization error bounds, among other PAC-Bayesian bounds. In this work, we show that the aforementioned framework
can also be adapted to obtain bounds based on a random subset of the dataset.
In particular:
\begin{enumerate}[label=(\roman*)]
\item We recover the bounds based on the individual sample mutual information~\citep[Proposition 1]{bu2020tightening} and on a random subset $S_J$ of the dataset~\citep[Theorem 2.4]{negrea2019information}.
These results are presented in Propositions~\ref{prop:bu_p1_recovery} and~\ref{prop:negrea_t24_recovery}.
\item We obtain new bounds based on the ``individual sample'' conditional mutual information $I(W; U_i|\tilde{Z}_i, \tilde{Z}_{i+N})$ and on a random subset $U_J$ of the Bernoulli trials.
These results are presented in Propositions~\ref{prop:icmi} and~\ref{prop:cond_version_negrea_t24}.
\item We further show that tighter bounds of the type of~\citep[Theorem 2.5]{negrea2019information} and~\citep[Theorem 3.7]{haghifam2020sharpened} cannot be recovered with this framework.
\end{enumerate}

Secondly, we develop expected generalization error bounds for the stochastic gradient Langevin dynamics (SGLD) algorithm. More specifically,
\begin{enumerate}[label=(\roman*)]
\setcounter{enumi}{3}
\item We extend~\citep[Theorem 3.7]{haghifam2020sharpened} in Proposition~\ref{prop:ext_haghifam_t43} in order to generalize the bound from~\citep[Theorem 4.2]{haghifam2020sharpened} for Langevin dynamics to SGLD in Proposition~\ref{prop:mini-batch-ld}, and to obtain the new bound in Proposition~\ref{prop:mini-batch-ld-mixture} using~\citep[Lemma 2]{rodriguez2020upper}.
Then, we combine Propositions~\ref{prop:mini-batch-ld} and~\ref{prop:mini-batch-ld-mixture} to generate our tightest bound on SGLD in Corollary~\ref{cor:mini-batch-ld}.
\end{enumerate}
\label{sec:introduction}

\section{Preliminaries}
\subsubsection{Notation}
Throughout this work, we write random variables in capital letters, $X$, and their realizations in lowercase letters, $x$.
We further write their target spaces in calligraphic letters, $\cX$, and the sigma algebras of their target space in script-style letters, $\mathscr{X}$.
We consider that all random variables are functions $X: \Omega \rightarrow \cX$ from an abstract probability space $(\Omega, \mathcal{A}, \bP)$ to a target space $(\cX, \mathscr{X})$.
Then, we denote their probability distribution by $P_X: \mathscr{X} \rightarrow [0,1]$, where $P_X(B) = \bP (X \in B)$ for all $B \in \mathscr{X}$.

When we consider more than one random variable, e.g., $X$ and $Y$, we write their 
joint probability distribution as $P_{X,Y}: \mathscr{X} \otimes \mathscr{Y} \rightarrow [0,1]$ and their product distribution as $P_{X} \times P_{Y}: \mathscr{X} \otimes \mathscr{Y} \rightarrow [0,1]$. 
Moreover, we also write the conditional distribution of $Y$ given $X$ as $P_{Y|X} : \mathscr{Y} \times \cX \rightarrow [0,1]$, which defines a probability distribution $P_{Y|X=x}$ over $(\mathcal{Y},\mathscr{Y})$ for each element $x \in \cX$.
Finally, we abuse notation and write $P_{Y|X} \times P_X = P_{X,Y}$ since $P_{X,Y}(B) =  \int \big( \int \chi_B \big((x,y) \big) dP_{Y|X=x}(y) \big) dP_X(x)$ for all $B \in \mathscr{X} \otimes \mathscr{Y}$, where $\chi_B$ is the characteristic function of the set $B$.

The relative entropy between two probability distributions $P$ and $Q$ is defined as
\[
	\KL{P}{Q} \triangleq \bE_P \left[\log \frac{dP}{dQ} \right],
\]
where $dP/dQ$ is the Radon--Nikodym derivative.
When one of the two distributions, $P$ or $Q$, depends on a random variable $R$, e.g., $P(R) = P_{X|R}$, then the relative entropy between the said distributions is also dependent on the random variable $R$.
For instance, in this example $\KL{P(R)}{Q} = f(R)$.
To say that $P$ is absolutely continuous w.r.t. $Q$ we write $P \ll Q$.

\vspace{1mm}
\subsubsection{A helpful Lemma}

We abstract and summarize the results from~\citep[Theorem 1 and Corollary 1, and Theorem 4 and Corollary 5]{hellstrm2020generalization} in the following Lemma.

\begin{lemma} 
\label{lemma:abstraction_hellstrom}
Let $f: \cX \times \mathcal{Y} \rightarrow \bR$ either:
\begin{enumerate}[label=(\roman*)]
	\item have zero mean and be $\sigma$-subgaussian under $P_X$ for all $y \in \mathcal{Y}$, or
	\item have zero mean and be $\sigma$-subgaussian under $P_X \times P_Y$.
\end{enumerate}  
Then, for all distributions $Q_Y$ over $(\mathcal{Y},\mathscr{Y})$,
\begin{equation}
	\big| \bE_{P_{X,Y}}[f(X,Y)] \big| \leq \sqrt{2\sigma^2 \KL{P_{X,Y}}{Q_Y \times P_X}}.
	\label{eq:abstraction_hellstrom}
\end{equation}
\end{lemma}

\label{sec:preliminaries}

\section{Subset Bounds}
In this section, we use Lemma~\ref{lemma:abstraction_hellstrom} to derive several expected generalization error bounds based on random subsets. We recover known bounds
for the standard setting and we also obtain their randomized subsample setting counterparts.

\subsection{Achievable subset bounds}

In the standard setting, the bounds using the individual sample mutual information, $I(W; Z_i)$, from~\citet[Proposition 1]{bu2020tightening} and the random subset from~\citet[Theorem 2.4]{negrea2019information} can be recovered using the framework from~\citep{hellstrm2020generalization} taking advantage of the subgaussianity of the empirical risk $L_{S_J}(W)$, of a subset $S_J$ of the dataset, under $P_{S_J}$. 

\begin{proposition}[{\cite[Proposition 1]{bu2020tightening}}]
\label{prop:bu_p1_recovery}
Under the assumptions of the standard setting, if $P_{W,Z_i} \ll P_{W} \times P_{Z}$ for all $i$ and either
\begin{enumerate}[label=(\roman*)]
    \item $\ell(w,Z)$ is $\sigma$-subgaussian under $P_Z$ for all $w \in \cW$ or 
    \item $\ell(W,Z)$ is $\sigma$-subgaussian under $P_{W} \times P_{Z}$,
\end{enumerate}
then the expected generalization error is bounded as follows:
\begin{equation}
    \big| \bE_{P_{W,S}} \big[ \textnormal{gen}(W,S) \big] \big| \leq \frac{1}{N} \sum_{i=1}^N \sqrt{2\sigma^2 I(W;Z_i)}.
\end{equation}
\end{proposition}
\begin{proof}[Sketch of the alternative proof]
In the setting of Lemma~\ref{lemma:abstraction_hellstrom}, let $f(X,Y) = \textnormal{gen}_i(W, Z_i) \triangleq \bE_{P_Z}[\ell(W,Z)] - \ell(W,Z_i)$, $P_X = P_{Z} = P_{Z_i}$, and $P_Y = Q_Y = P_{W}$. Then, we have that
\begin{equation*}
	\left| \bE_{P_{W,\smash{Z_i}}} \big[ \textnormal{gen}_i(W,Z_i) \big] \right| \leq \sqrt{2\sigma^2 I(W;Z_i)}.
\end{equation*}
Finally, we note that $\bE_{P_{W,S}}[\ell(W,Z_i)] = \bE_{P_{W,Z_i}}[\ell(W,Z_i)]$ and we use the triangle inequality to obtain the desired result.
\end{proof}

\begin{proposition}[{Extension of \citep[Theorem 2.4]{negrea2019information}}]
\label{prop:negrea_t24_recovery}
Consider the assumptions of the standard setting.
Also consider a random subset $J \subseteq [N]$ such that $|J| = M$, which is uniformly distributed and independent of $W$ and $S$, and a random variable $R$ which is independent of $S$ and $J$.
If $P_{W,S_J|S_{J^c},R} \ll Q_{W|S_{J^c},R} \times P_{S_J}$ for all $J$,\footnote{We note that $Q_{W|S_{J^c},R}$ might differ from $P_{W|S_{J^c},R}$, the marginal of $P_{W,S_J|S_{J^c},R}$ with respect to $W$.} and either:
\begin{enumerate}[label=(\roman*)]
    \item $\ell(w,Z)$ is $\sigma$-subgaussian under $P_Z$ for all $w \in \cW$ or 
    \item $\ell(W,Z)$ is $\sigma$-subgaussian under $P_{W|S_{J^c},R} \times P_Z$,
\end{enumerate}
then the expected generalization error is bounded as follows:
\begin{align}
\MoveEqLeft[.8]
 \big| \bE_{P_{W,S}} \big[ \textnormal{gen}(W,S) \big] \big| \nonumber \\
 &\leq \bE \Bigg[ \sqrt{\frac{2\sigma^2}{M} \KL[\big]{P_{W,S_J |S_{J^c},R }}{Q_{W|S_{J^c},R } \times P_{S_J}}} \Bigg],
 \label{eq:negrea_t24_recovery}
\end{align}
where $Q_{W|S_{J^c},R}$ is a $(\mathscr{Z}^{\otimes (N - M)} \otimes \mathscr{R} \otimes \mathscr{J})$-measurable distribution over $(\cW,\mathscr{W})$, and the expectation on the r.h.s. of~\eqref{eq:negrea_t24_recovery} is over $P_{ J,S_{J^c},R }$.
\end{proposition}
\begin{proof}[Sketch of the alternative proof]
In the setting of Lemma~\ref{lemma:abstraction_hellstrom} and assuming a fixed $J$, let $f(X,Y) = \textnormal{gen}_J(W, S_J) \triangleq \bE_{P_Z}[\ell(W, Z)] - \frac{1}{M} \sum_{i \in J} \ell(W,Z_i)$, $P_X = P_{S_J}$, and $P_Y = P_{W|S_{J^c},R}$. Then, since $\textnormal{gen}_J$ is $\sigma/\sqrt{M}$-subgaussian and zero-mean under $P_{S_J}$, we have that,
\begin{multline*}
 \big| \bE_{P_{W,S_J|S_{J^c},R}} \big[ \textnormal{gen}_J(W,S_J) \big] \big| \\
 \leq \sqrt{\frac{2\sigma^2}{M} \KL[\big]{P_{W,S_J|S_{J^c},R}}{Q_{W|S_{J^c},R} \times P_{S_J}}}.
\end{multline*}
Then, by Jensen's inequality, $| \bE_{P_{W,S,R,J}} [ \textnormal{gen}_J(W,S_J) ] | \leq \bE_{P_{J,S_{J^c},R}} [ | \bE_{P_{W,S_J|S_{J^c},R }} [ \textnormal{gen}_J(W,S_J) ] | ]$, which
completes the proof of the Proposition.
%
\end{proof}

\begin{remark}
Note that in \citep[Theorem 2.4]{negrea2019information} the result is formulated in terms of the complementary random subset $J \leftarrow J^c$; we use this other formulation instead to maintain the notation throughout the paper.
Moreover, we prove that~\eqref{eq:negrea_t24_recovery} also holds for condition (ii), which is not implied by condition (i)~\citep[Appendix C]{negrea2019information}, and also avoid forfeiting the absolute value.
\end{remark}

In the sequel, we find a new bound based on the individual sample conditional mutual information $I(W;U_i|\tilde{Z}_i,\tilde{Z}_{i+N})$, which is the tightest mutual-information--based bound in this setting (see Appendix~\ref{app:individual_cmi_tightest}
), and on random subsets of the data.
For this task, we must take advantage of the subgaussianity of the difference of the empirical risks $L_{\smash{\bar{S}}_J}(W) - L_{S_J}(W)$ of subsets (indexed by $J$) of the dataset $S$ and the rest of the super-sample $\bar{S} = \tilde{S} \setminus S$, under $P_{U_J}$.

\begin{proposition}[Individual conditional mutual information bound]
\label{prop:icmi}
Under the assumptions of the randomized subsample setting, if $P_{\smash{W,\tilde{Z}_i,\tilde{Z}_{i+N},U_i}} \ll P_{\smash{W,\tilde{Z}_i,\tilde{Z}_{i+N}}} \times P_{U_i}$ for all $i$ and $\ell(w,z)$ is bounded in $[a,b]$ for all $w \in \cW$ and $z \in \cZ$,
then the expected generalization error is bounded as follows:
\begin{multline}
    \left| \bE_{P_{W,S}} \big[ \textnormal{gen}(W,S) \big] \right| \\
    \leq \frac{1}{N} \sum_{i=1}^N \sqrt{2(b-a)^2 I(W;U_i|\tilde{Z}_i, \tilde{Z}_{i+N})}.
    \label{eq:icmi_result}
\end{multline}
\end{proposition}
\begin{proof}[Sketch of the proof]
In the setting of Lemma~\ref{lemma:abstraction_hellstrom}, let $f(X,Y) = \widehat{\textnormal{gen}}_i \big( W, \tilde{Z}_i, \tilde{Z}_{i+N}, U_i \big) \triangleq \ell \big( W, \tilde{Z}_{i + (1-U_i)N} \big) - \ell \big( W, \tilde{Z}_{i+U_iN} \big)$, $P_X = P_{U_i}$, and $P_Y = Q_Y = P_{\smash{W,\tilde{Z}_i,\tilde{Z}_{i+N}}}$.
Then, since $\widehat{\textnormal{gen}}_i$ is zero-mean and $(b-a)$-subgaussian (since it is bounded in $[a-b,b-a]$) under $P_{U_i}$, we have that
\begin{multline*}
 \big| \bE_{P_{\smash{W, \tilde{Z}_i, \tilde{Z}_{i+N}, U_i}}}  \big[ \widehat{\textnormal{gen}}_i \big( W, \tilde{Z}_i, \tilde{Z}_{i+N}, U_i \big) \big]  \big| \\
 \leq \sqrt{2(b-a)^2 I(W;U_i|\tilde{Z}_i,\tilde{Z}_{i+N})}.
\end{multline*}
If we use that $\big| \sum_{i=1}^N x_i \big| \leq \sum_{i=1}^N |x_i|$ and note that 
\begin{multline*}
 \bE_{P_{W,\smash{\tilde{S}},U}} \big[ \widehat{\textnormal{gen}} \big( W, \tilde{S}, U \big) \big] \\
 = \frac{1}{N} \sum_{i=1}^N \bE_{P_{\smash{W,\tilde{Z}_i,\tilde{Z}_{i+N},U_i}}}  \big[ \widehat{\textnormal{gen}}_i \big( W, \tilde{Z}_i, \tilde{Z}_{i+N},  U_i \big) \big],
\end{multline*}
we obtain the bound~\eqref{eq:icmi_result}.
\end{proof}

\begin{proposition}
\label{prop:cond_version_negrea_t24}
Consider the assumptions of the randomized subsample setting.
Also consider a subset $J \subseteq [N]$ such that $|J| = M$, which is uniformly distributed and independent of $W$, $\tilde{S}$, and $U$, and a random variable $R$ which is independent of $\tilde{S}$, $J$, and $U$.
If $P_{\smash{W,U_J|U_{J^c},\tilde{S},R}} \ll Q \times P_{U_J}$ for all $J$ and $\ell(w,z)$ is bounded in $[a,b]$ for all $w \in \cW$ and $z \in \cZ$,
then the expected generalization error is bounded as follows:
\begin{align}
\MoveEqLeft[1]
 \big| \bE_{P_{W,S}} \big[ \textnormal{gen}(W,S) \big] \big| \nonumber\\
 &\leq \bE \bigg[ \sqrt{ \frac{2(b-a)^2}{M}\KL[\big]{ P_{\smash{W,U_J|U_{J^c},\tilde{S},R}} }{ Q \times P_{U_J} } } \bigg],
 \label{eq:cond_version_negrea_t24}
\end{align}
where $Q=Q_{\smash{W|U_{J^c},\tilde{S},R}}$ is a $(\mathscr{U}^{\otimes (N - M)} \otimes \mathscr{Z}^{\otimes 2N} \otimes \mathscr{R} \otimes \mathscr{J})$-measurable distribution over $(\cW,\mathscr{W})$, and the expectation on the r.h.s. of~\eqref{eq:cond_version_negrea_t24} is with respect to $P_{\smash{J, U_{J^c}, \tilde{S}, R}}$.
\end{proposition}
\begin{proof}[Sketch of the proof]
In the setting of Lemma~\ref{lemma:abstraction_hellstrom} and assuming a fixed $J$, let $f(X,Y) = \widehat{\textnormal{gen}}_J \big( W, U_J, \tilde{S}_J, \tilde{S}_{J+N} \big) \triangleq \frac{1}{M} \sum_{i \in J} \big[ \ell \big( W, \tilde{Z}_{i+(1-U_i)N} \big) - \ell \big( W, \tilde{Z}_{i+U_iN} \big) \big]$, $P_X = P_{U_J}$, and $P_Y = P_{\smash{W|U_{J^c}, \tilde{S}, R}}$.
Since each of the $M$ summands of $\widehat{\textnormal{gen}}_J$ is zero-mean and bounded in $[a-b,b-a]$ under $P_{U_J}$, then $\widehat{\textnormal{gen}}_J$ is zero-mean and $(b-a)/\sqrt{M}$-subgaussian under $P_{U_J}$.
Therefore, we have that
\begin{align*}
\MoveEqLeft[.9]
 \big| \bE_{P_{\smash{W,U_J| U_{J^c},\tilde{S},R}}} \big[ \widehat{\textnormal{gen}}_J \big( W, U_J, \tilde{S}_J, \tilde{S}_{J+N} \big) \big] \big| \nonumber\\
 &\leq \sqrt{\frac{2(b-a)^2}{M} \KL[\big]{P_{\smash{W,U_J| U_{J^c},\tilde{S},R}}}{Q_{\smash{W| U_{J^c},\tilde{S},R}} \times P_{U_J}}}.
\end{align*}
By Jensen's inequality, $|\bE_{P_{\smash{W,\tilde{S},U,R,J}}}[\widehat{\textnormal{gen}}_J(W,U_J\tilde{S}_J,\tilde{S}_{J+N})]| \allowbreak \leq \bE_{P_{\smash{J,U_{J^c},\tilde{S},R}}} [| \bE_{P_{\smash{W,U_J| U_{J^c},\tilde{S},R}}} [ \widehat{\textnormal{gen}}_J ( W, U_J, \tilde{S}_J, \tilde{S}_{J+N} ) ] |]$ and the proof of the Proposition is complete.
%
\end{proof}

\begin{remark}
Note that a similar bound can be directly obtained combining a slight modification of~\citep[Theorem 3.1]{haghifam2020sharpened} (where $R$ is included) and~\citep[Lemma 3.6]{haghifam2020sharpened}.
However, this approach results in a bound without the absolute value.
\end{remark}

\subsection{Non-achievable subset bounds}

Despite its versatility, the present framework does not allow us to obtain bounds where all the expectations are outside of the square root, such as~\citep[Theorem~2.5]{negrea2019information} or~\citep[Theorem~3.7]{haghifam2020sharpened}.
If we look at~\eqref{eq:abstraction_hellstrom}, we see that the term $\KL{P_{X,Y}}{Q_Y \times P_X}$, which can also be written as $\bE_{P_X}[\KL{P_{Y|X}}{Q_Y}]$, is inside the square root; hence, Jensen's inequality prevents us from taking the expectation with respect to $P_X$ ($P_{S_J}$ or $P_{U_J}$ in our case) outside the concave square root function.

In particular, a tighter (when $M=1$) version of~\citep[Theorem~2.4]{negrea2019information} is found in~\citep[Theorem 2.5]{negrea2019information}, with the same conditions as Proposition~\ref{prop:negrea_t24_recovery}, except that $\ell(w,z)$ is now required to be bounded in $[a,b]$; namely
\begin{multline}
 \bE_{P_{W,S}} \big[ \textnormal{gen}(W,S) \big] \\ \leq \frac{b-a}{\sqrt{2}} \
 \bE_{P_{ J,S,R }} \bigg[ \sqrt{ \KL[\big]{P_{W |S,R }}{Q_{ \smash{W|S_{J^c},R} }}} \bigg].
 \label{eq:negrea_t25}
\end{multline}

In the same spirit, a tighter version of Proposition~\ref{prop:cond_version_negrea_t24} is~\citep[Theorem~3.7]{haghifam2020sharpened}, from which we present a slight extension, as it will be needed for the bounds on the stochastic gradient Langevin dynamics algorithm in the following section.

\begin{proposition}[{Extension of~\citep[Theorem~3.7]{haghifam2020sharpened}}]
\label{prop:ext_haghifam_t43}
Consider the assumptions of the randomized subsample setting.
Also consider a subset $J \subseteq [N]$ such that $|J| = M$, which is uniformly distributed and independent of $W$, $\tilde{S}$, and $U$, and a random variable $R$ which is independent of $\tilde{S}$, $J$, and $U$.
If $P_{\smash{W|U,\tilde{S},R}} \ll Q_{\smash{W|U_{J^c},\tilde{S},R}}$ for all $J$ and $\ell(w,z)$ is bounded in $[a,b]$ for all $w \in \cW$ and $z \in \cZ$,
then the expected generalization error is bounded as follows:
\begin{multline}
 \bE_{P_{W,S}} \big[ \textnormal{gen}(W,S) \big] \leq \sqrt{2}(b-a)\\ \bE_{P_{\smash{J,\tilde{S},U,R}}} \bigg[ \sqrt{ \KL[\big]{P_{\smash{W|U,\tilde{S},R}}}{Q_{\smash{W|U_{J^c},\tilde{S},R}}}} \bigg],
 \label{eq:ext_haghifam_t43}
\end{multline}
where $Q_{\smash{W|U_{J^c},\tilde{S},R}}$ is a $(\mathscr{U}^{\smash{\otimes (N - M)}} \otimes \mathscr{Z}^{\otimes 2N} \otimes \mathscr{R} \otimes \mathscr{J})$-measurable distribution over $(\cW,\mathscr{W})$.
\end{proposition}

\begin{proof}
The extension follows trivially from the proof of Theorem~3.7 on~\citep{haghifam2020sharpened}.
The random variable $R$ is incorporated using the Donsker--Varadhan variational formula~\citep[Theorem 3.5]{polyanskiy2014lecture} with the probability distributions
$P(\tilde{S},U,R)$ and $Q(\tilde{S},U_J,J,R)$ instead of $P(\tilde{S},U)$ and $Q(\tilde{S},U_J,J)$.
Then, the term $(b-a)$ results from the boundedness of the Donsker--Varadhan variational function.
Finally, the extension to arbitrary sizes of the subset $J$ is already included in the original proof in~\citep[Appendix D]{haghifam2020sharpened}.
\end{proof}

\begin{remark}
In Propositions~\ref{prop:negrea_t24_recovery} and~\ref{prop:cond_version_negrea_t24}, as well as in~\eqref{eq:negrea_t25} and Proposition~\ref{prop:ext_haghifam_t43}, when $M=1$, similarly to Propositions~\ref{prop:bu_p1_recovery} and~\ref{prop:icmi}, the outermost expectation with respect to $J$ is equivalent to a sum over all the indices $i \in [N]$ divided by $N$.
That is, $\bE_{J}[f(J)] = \frac{1}{N} \sum_{i=1}^N f(i)$.
\end{remark}
\label{sec:subset_bounds}

\section{Bounds on Noisy Iterative Algorithms}
In this section, we leverage the bound from Proposition~\ref{prop:ext_haghifam_t43}, 
not based on Lemma~\ref{lemma:abstraction_hellstrom}, to extend the result on Langevin dynamics from~\citep[Theorem~4.2]{haghifam2020sharpened} to stochastic gradient Langevin dynamics.
After that, we propose an alternative for such extension for algorithms with large Lipschitz constant or large discrepancy between the gradients of two samples.
Finally, we combine both bounds to generate a refined version that better 
exploits the 
samples' and hypotheses' trajectory knowledge.

\subsection{Stochastic gradient Langevin dynamics algorithm}

The stochastic gradient Langevin dynamics (SGLD) algorithm is an iterative procedure to learn a parametrized hypothesis $W_{\theta}$ from a dataset $S$.
More specifically, SGLD works for hypotheses that are completely characterized by a parameter $\theta \in \bR^d$.
The algorithm starts with a random initialization, $\theta_0$, of the parameter.
Then, at each iteration $t \in [T]$, it samples a random batch $S_{V_t}$ from the dataset $S$; updates the previous parameter $\theta_{t-1}$ with a scaled $(-\eta_t)$ version of the gradient of the loss given that parameter and the random batch, i.e., $\nabla_{\theta_{t-1}} L_{S_{V_t}} (W_{\theta_{t-1}})$; and adds a scaled $(\sigma_t)$ isotropic Gaussian random noise $\varepsilon_t \sim \cN(0,I_d)$. That is, 
\begin{equation*}
	\theta_t \leftarrow \theta_{t-1} - \eta_t \nabla_{\theta_{t-1}} L_{S_{V_t}}(W_{\theta_{t-1}}) + \sigma_t \varepsilon_t,
\end{equation*}
where $W_{\theta_t}$ (or $W_t$) is the hypothesis at iteration $t$, and $W_{\theta_T}$ is the final hypothesis of SGLD.
When the batch is composed of all samples, i.e., there is no stochasticity in the sample selection, the algorithm is called 
Langevin dynamics (LD).

\subsection{Expected generalization error bounds}

\begin{proposition}[{Extension of~\citep[Theorem~4.2]{haghifam2020sharpened} to SGLD}]
\label{prop:mini-batch-ld}
Under the conditions of the randomized subsample setting, if we assume that the loss function $\ell(w,z)$ is bounded in $[a,b]$ for all $w \in \cW$ and all $z \in \cZ$, then the expected generalization error of the SGLD is bounded as follows:
\begin{multline}
 \bE_{P_{W,S}} \big[ \textnormal{gen}(W,S) \big] \leq \sqrt{2} (b-a) \\
 \bE_{\smash{P_1}} \Bigg[ \sqrt{ \sum_{\smash{t \in \mathcal{T}_{J}(V^T)}} \bE_{\smash{P_2^t}} \bigg[ \frac{\eta_t^2 \|\zeta_{J,t}\|^2}{2\sigma_t^2|V_t|^2} (U_J - \pi_{J,t})^2 \bigg] } \Bigg],
 \label{eq:mini-batch-ld}
\end{multline}
where $P_1 = P_{\smash{J, \tilde{S}, U, V^T}}$, $P_2^t= P_{\smash{W^{t-1}| U, \tilde{S}, V^{t-1}}}$, $|J|=1$, $|V_t|$ is the cardinality of the batch, $\mathcal{T}_J \big( V^T \big)$ is the set of iterations for which sample $J$ was included in the batches from $V^T$, $\zeta_{J,t} \triangleq \nabla_{\theta_{t-1}} \ell( W_{\smash{\theta_{t-1}}}, \tilde{Z}_{J}) - \nabla_{\theta_{t-1}} \ell( W_{\smash{\theta_{t-1}}}, \tilde{Z}_{J+N})$ is the two-sample incoherence at iteration $t$, and $\pi_{J,t}$ is an estimate (based on the samples' and hypotheses' trajectory) of the probability that $U_J = 1$.
\end{proposition}
\begin{proof}[Sketch of the proof]
Similarly to the proof of Theorem~4.2 in~\citep{haghifam2020sharpened}, we start from the bound in Proposition~\ref{prop:ext_haghifam_t43}, where here we let $R$ be the full batch trajectory $V^T$, and we use the fact that $\KL[\big]{P_{\smash{W|U,\tilde{S},V^T}}}{Q_{\smash{W|U_{J^c},\tilde{S},V^T}}} \leq \sum_{t=1}^T \allowbreak \bE_{\smash{P_2^t}} \big[ \KL[\big]{P_{\smash{W_t| W_{t-1}, U, \tilde{S}, V^t}}}{Q_{\smash{W| W^{t-1}, U_{J^c}, \tilde{S}, V^t}}} \big]$.
Then, as in~\citep{bu2020tightening}, we restrict the sum to only the non-zero terms, that is $\sum_{t=1}^T \bE_{\smash{P_2^t}} \big[ \KL[\big]{P_{\smash{W_t| W_{t-1}, U, \tilde{S}, V^t}}}{Q_{\smash{W| W^{t-1}, U_{J^c}, \tilde{S}, V^t}}} \big] = \sum_{t \in \mathcal{T}_J(V^T)}\! \bE_{\smash{P_2^t}}\! \big[  \KL[\big]{P_{\smash{W_t| W_{t-1}, U, \tilde{S}, V^t}}}{Q_{\smash{W| W^{t-1}, U_{J^c}, \tilde{S}, V^t}}} \!\big]$.
After that, similarly to~\citep{haghifam2020sharpened}, we note that $P_{\smash{W_t|W_{t-1}, U, \tilde{S}, V^t}}$ is a Gaussian distribution and we then let $Q_{\smash{W|W^{t-1}, U_{J^c}, \tilde{S}, V^t}}$ be also a Gaussian distribution, but with a different mean.
The first distribution uses $U_J$ to determine which gradient $\nabla_{\theta_{t-1}} \ell(W_{\theta_{t-1}},\tilde{Z}_{J+U_J N})$ appears on its mean, while the latter uses a weighted average of both gradients by means of $\pi_{J,t}$.
Finally, we use the analytical expression for the relative entropy between two Gaussian distributions to obtain~\eqref{eq:mini-batch-ld}.
\end{proof}

\begin{figure*}[!t]
\normalsize
\setcounter{MYtempeqncnt}{\value{equation}}
\setcounter{equation}{9}
\begin{equation}
\bE_{P_{W,S}} \big[ \textnormal{gen}(W,S) \big]  \leq  \sqrt{2}(b-a) \, \bE_{P_1} \smash{ \Bigg[  \sqrt{ \sum\nolimits_{t \in \mathcal{T}_{J}(V^T)} \bE_{P_2^t} \bigg[ {-} \log \bigg( \big| U_J - \pi_{J,t} \big| \exp \Big( - \frac{\eta_t^2 \|\zeta_{J,t}\|^2}{2\sigma_t^2|V_t|^2} \Big) + \big| \bar{\pi}_{J,t} - U_J \big| \bigg) \bigg] }  \Bigg] }
 \label{eq:mini-batch-ld-mixture}
\end{equation}
\vspace*{-5pt}
\setcounter{equation}{\value{MYtempeqncnt}}
\hrulefill
\vspace*{-5pt}
\end{figure*}

\begin{corollary} 
\label{cor:mini-batch-ld-lipschitz}
Consider the setting of Proposition~\ref{prop:mini-batch-ld}.
In the case that $\ell$ is $L$-Lipschitz, and if we assume that the batches have a constant size $K$, the generalization error of the SGLD is bounded as follows:
\begin{multline}
 \bE_{P_{W,S}} \big[ \textnormal{gen}(W,S) \big] \leq \frac{2L}{K}(b-a) \\
 \bE_{P_1} \Bigg[ \sqrt{\sum_{\smash{t \in \mathcal{T}_{J}(V^T)}} \!\! \bE_{P_2^t} \bigg[ \frac{\eta_t^2}{\sigma_t^2} (U_J - \pi_{J,t})^2 \bigg]} \Bigg].
    \label{eq:mini-batch-ld-lipschitz}
\end{multline}
\end{corollary}

The bound from Proposition~\ref{prop:mini-batch-ld} recovers~\citep[Theorem~4.2]{haghifam2020sharpened} for LD (i.e., $|V_t| = N$) and, for SGLD, it can be compared with other bounds.
For instance, under the assumptions from Corollary~\ref{cor:mini-batch-ld-lipschitz}, if we compare it with~\citep[Theorem 3.1]{negrea2019information}:
\[
 \bE_{P_{W,S}} \big[ \textnormal{gen}(W,S) \big] \leq \frac{L(b-a)}{\sqrt{2}K} \,
 \bE_{P_{\smash{J, V^T}}} \Bigg[ \sqrt{\sum_{\smash{t \in \mathcal{T}_{J}(V^T)}} \! \frac{\eta_t^2}{\sigma_t^2}} \Bigg],
\]
we note that~\eqref{eq:mini-batch-ld-lipschitz} has a worse constant.
However, if $\bE_{P_2^t} \big[ (U_J - \pi_{J,t})^2 \big] \leq 1/8$, the bound from Corollary~\ref{cor:mini-batch-ld-lipschitz} is tighter.
This is expected to happen after some iterations, when a good estimate of $U_J$ is possible.
The same happens when we assume $M=1$ and compare it with~\citep[Proposition 3]{bu2020tightening}, where $\bE_{P_2^t} \big[ (U_J - \pi_{J,t})^2 \big] \leq 1/16$ is required.

\begin{remark}
The estimate $\pi_{J,t}$ is dependent  on the samples' and hypotheses' trajectory.
An example on how to build such an estimate, based on binary hypothesis testing, is presented in~\citep{haghifam2020sharpened} for LD.
Adapted to SGLD, we let the estimate $\pi_{J,t}$ be a function $\phi: \bR \rightarrow [0,1]$ of the log-likelihood ratio between the probability that $U_J = 1$ and $U_J = 0$, based on $(W^{t-1},J,\tilde{S},U_{J^c},V^{t})$. That is, 
\begin{align*}
 \pi_{J,t} 
 &\triangleq \phi \left( \log \frac{P_{\smash{U_J| W^{t-1}, \tilde{S}, U_{J^c},V^{t}}}(1)}{P_{\smash{U_J| W^{t-1}, \tilde{S}, U_{J^c},V^{t}}}(0)} \right) \\
 &= \phi \left( \sum\nolimits_{t \in T_J(V^{t})} \big( Y_{J,t,0} - Y_{J,t,1} \big) \right),
\end{align*}
where
\begin{align*}
\MoveEqLeft[6]
 Y_{J,t,u} = \frac{1}{2 \sigma_t^2} \Big\| \theta_t {-} \theta_{t-1} {+} \frac{\eta_t}{|V_{t}|} \Big( \nabla_{\smash{\theta_{t-1}}} \ell \big( W_{\smash{\theta_{t-1}}}, \tilde{Z}_{J + u N} \big)  \\ 
 &+ (|V_{t}|-1) \nabla_{\smash{\theta_{t-1}}} L_{S_{V_{t} \setminus J}}(W_{\theta_{t-1}}) \Big) \Big\|^2.
\end{align*}
\end{remark}

A potential weakness in the bound from Proposition~\ref{prop:mini-batch-ld} is its linear dependence with $\|\zeta_{J,t}\|$, which might make the bound loose if the Lipschitz constant $L$ is large.
As a remedy, we propose the following alternative.

\addtocounter{equation}{1}

\begin{proposition}
\label{prop:mini-batch-ld-mixture}
Under the setting of Proposition~\ref{prop:mini-batch-ld}, the expected generalization error of SGLD is bounded as
shown in~\eqref{eq:mini-batch-ld-mixture} at the top of the page,
where $\bar{\pi}_{J,t} = 1 -\pi_{J,t}$.
\end{proposition}
\begin{proof}[Sketch of the proof]
We follow the steps of the proof of Proposition~\ref{prop:mini-batch-ld}, but instead of considering $Q_{\smash{W| W^{t-1}, U_{J^c}, \tilde{S}, V^t}}$ as a single Gaussian distribution, we define it as a mixture of two Gaussians.
The responsibility of each mixture component is $\pi_{J,t}$ and $\bar{\pi}_{J,t}$ with either $\nabla_{\theta_{t-1}} \ell \big( W_{\smash{\theta_{t-1}}}, \tilde{Z}_{J+N} \big)$ or $\nabla_{\theta_{t-1}} \ell \big( W_{\smash{\theta_{t-1}}}, \tilde{Z}_{J} \big)$ appearing on their respective mean.
Then, we apply~\citep[Lemma~2]{rodriguez2020upper} to bound the relative entropy $\KL[\big]{P_{\smash{W_t| W_{t-1}, U, \tilde{S}, V^t}}}{Q_{\smash{W| W^{t-1}, U_{J^c},\tilde{S}, V^t}}}$, and the rest follows by algebraic manipulation.
\end{proof}

For example, if $\pi_{J,t} = 1/2$, which is a reasonable value during the first iterations, the bound from Proposition~\ref{prop:mini-batch-ld-mixture} is tighter than the one from Proposition~\ref{prop:mini-batch-ld} for $\| \zeta_{J,t} \| \gtrsim 2.21\sigma_t |V_t|/\eta_t$; this may likely be the case when $\eta_t / \sigma_t \in \Theta(N^{\alpha/2})$, for $\alpha \in (0,1)$, see e.g. \citep[Appendix E]{negrea2019information}.
In fact, for large values of ${\frac{\eta_t^2 \|\zeta_{J,t}\|^2}{2\sigma_t^2|V_t|^2}}$, e.g., when the Lipschitz condition is not met, the term inside the innermost expectation in~\eqref{eq:mini-batch-ld-mixture} tends to $ {-}\log|\bar{\pi}_{J,t}-U_J|$ from below.

Moreover, we note that the only terms that change between~\eqref{eq:mini-batch-ld} and~\eqref{eq:mini-batch-ld-mixture} are the summands inside the square root.
These terms are obtained by considering different data-dependent distributions $Q_{\smash{W_t| W^{t-1}, U_{J^c}, \tilde{S}, V^t}}$, which we can choose arbitrarily at each iteration; hence, we may choose the tightest form for each summand, as it is stated next.

\begin{corollary}
\label{cor:mini-batch-ld}
Under the setting of Proposition~\ref{prop:mini-batch-ld}, if we let
\begin{align*}
f_{J,t} &\triangleq \frac{\eta_t^2 \|\zeta_{J,t}\|^2}{2\sigma_t^2|V_t|^2} (U_J - \pi_{J,t})^2 \textnormal{ and } \\
g_{J,t} &\triangleq -\log \Big(|U_J - \pi_{J,t}| \smash{e^{- \frac{\eta_t^2 \|\zeta_{J,t}\|^2}{2\sigma_t^2|V_t|^2}}} + |\bar{\pi}_{J,t} - U_J| \Big),
\end{align*}
then the expected generalization error of SGLD is bounded as follows: 
\begin{multline}
 \bE_{P_{W,S}} \big[ \textnormal{gen}(W,S) \big] \leq \sqrt{2} (b-a) \\
 \bE_{\smash{P_1}} \bigg[ \sqrt{ \sum\nolimits_{t \in \mathcal{T}_{J}(V^T)} \bE_{\smash{P_2^t}} \Big[ \min \big( f_{J,t},\, g_{J,t} \big) \Big] } \bigg].
\end{multline}
\end{corollary}

\label{sec:noisy_iterative_algorithms_bounds}

\section{Discussion}
In this paper, we showed how the framework from~\citep{hellstrm2020generalization} allows us to recover the individual sample and subset bounds from~\citep{bu2020tightening} and~\citep{negrea2019information}, as well as to generate their parallels in the randomized subsample setting from \citep{steinke2020reasoning}.
Moreover, we showed why the framework does not allow us to obtain bounds with all the expectations outside the square root, such as the tightest bounds from~\citep{negrea2019information} and~\citep{haghifam2020sharpened}.
Finally, we extended the LD bounds from~\citep{haghifam2020sharpened} to SGLD and we refined them to loss functions with potentially large gradient norms.
\label{sec:discussion}

\bibliographystyle{IEEEtranN}
\bibliography{references}

\newpage
\onecolumn
\appendix
\subsection{Details of Lemma~\ref{lemma:abstraction_hellstrom}}

Before starting with the proof of the Lemma, we introduce the \emph{information density} between two random variables $X$ and $Y$, which is defined as
\[
\iota(X;Y) \triangleq \log \frac{dP_{X,Y}}{d(P_X \times P_Y)},
\]
and the \emph{conditional information density} between two random variables $X$ and $Y$ given a random variable $Z$, which is defined as 
\[
\iota(X;Y|Z) \triangleq \log \frac{dP_{X,Y|Z}}{d(P_{X|Z} \times P_{Y|Z})}.
\]

Note that if we consider two random variables $X$ and $Y|Z$, where $Y|Z$ is a random variable described by the conditional probability distribution $P_{Y|Z}$, and where $X$ is independent of $Z$, then their information density $\iota(X;(Y|Z))$ is equal to the conditional information density between $X$ and $Y$ given $Z$, $\iota(X;Y|Z)$, i.e.,
\[
\log \frac{dP_{X,Y|Z}}{d(P_{Y|Z}\times P_X)}.
\]

\begin{proof}[Proof of Lemma~\ref{lemma:abstraction_hellstrom}]
If (i) holds, i.e., $f(X,y)$ is $\sigma$-subgaussian under $P_X$ for all $y \in \mathcal{Y}$ and we know that $\bE_{P_X} \big[ f(X,y) \big] = 0$ for all $y \in \mathcal{Y}$, then we have that 
\begin{equation*}
    \bE_{P_X} \big[ \exp \big( \lambda f(X,y) \big) \big] \leq \exp \left( \frac{\lambda^2 \sigma^2}{2} \right),
\end{equation*}
for all $y \in \mathcal{Y}$.
Then, we can rearrange the terms and take the expectation with respect to $P_Y$ to obtain
\begin{equation*}
    \bE_{P_Y \times P_X} \left[ \exp \left(\lambda f(X,Y) - \frac{\lambda^2 \sigma^2}{2} \right) \right] \leq 1,
\end{equation*}
which is the same expression we would obtain with the subgaussianity and zero-mean conditions imposed in (ii).
Now, let us consider the support of $P_{X,Y}$ to be $A = \textnormal{supp}(P_{X,Y})$.
Then, it follows that
\begin{equation*}
    \bE_{P_Y \times P_X} \left[ \chi_{A} \exp \left(\lambda f(X,Y) - \frac{\lambda^2 \sigma^2}{2} \right) \right] \leq 1,
\end{equation*}
where $\chi_A$ is the characteristic function of the collection of sets $A$.
Finally, if we employ~\citep[Proposition 17.1]{polyanskiy2014lecture} we have that
\begin{equation*}
    \bE_{P_{X,Y}} \left[ \exp \left(\lambda f(X,Y) - \frac{\lambda^2 \sigma^2}{2} - \iota(X;Y) \right) \right] \leq 1.
\end{equation*}

Now, we can insert the expectation over $P_{X,Y}$ inside the exponential by means of Jensen's inequality to obtain
\begin{equation*}
    \exp \left( \lambda\, \bE_{P_{X,Y}} [ f(X,Y) ] - \frac{\lambda^2 \sigma^2}{2} - \KL{P_{X,Y}}{P_Y \times P_X} \right)  \leq 1.
\end{equation*}
Then, if we optimize for $\lambda$ and rearrange the terms we obtain that
\begin{equation*}
    \big| \bE_{P_{X,Y}}[f(X,Y)] \big| \leq \sqrt{2 \sigma^2 \KL{P_{X,Y}}{P_Y \times P_X}}.
\end{equation*}
Finally, we apply~\citep[Corollary 3.1]{polyanskiy2014lecture} to substitute $P_{Y}$ for any distribution $Q_Y$ over $(\mathcal{Y},\mathscr{Y})$ and complete the proof. 
\end{proof}

Note that we can apply Lemma~\ref{lemma:abstraction_hellstrom} to a function $f: \mathcal{X} \times (\mathcal{Y} \times \mathcal{Z}) \rightarrow \bR$ and consider the random variables $X$, which is described by $P_{X}$ and is independent of $Z$, and $Y|Z$, which is described by $P_{Y|Z}$.
Then, if $f$ is either:
\begin{enumerate}[label=(\roman*)]
    \item zero-mean and $\sigma$-subgaussian under $P_X$ for all $y \in \mathcal{Y}$, or
    \item zero-mean and $\sigma$-subgaussian under $P_{Y|Z} \times P_{X}$,
\end{enumerate}
then, for all $\mathscr{Z}$-measurable distributions $Q_{Y|Z}$ over $(\mathcal{Y},\mathscr{Y})$, we have that
\begin{equation*}
    \big| \bE_{P_{X,Y|Z}}[f(X,Y,Z)] \big| \leq \sqrt{2 \sigma^2 \KL{P_{X,Y|Z}}{Q_{Y|Z} \times P_X}}.
\end{equation*}

\subsection{Proof of Proposition~\ref{prop:bu_p1_recovery}}

Consider the setting of Lemma~\ref{lemma:abstraction_hellstrom} and let $X=Z_i$, $Y = W$, and $f(X,Y) = \textnormal{gen}_i(W,Z_i) \triangleq \bE_{P_Z}[\ell(w,Z)] - \ell(w,Z_i)$.
Then, since $\textnormal{gen}_i(W,Z_i)$ is either (i) zero-mean and $\sigma$-subgaussian under $P_{Z_i}$ for all $w \in \cW$ or (ii) zero-mean and $\sigma$-subgaussian under $P_W \times P_{Z_i}$, we have that 
\begin{equation*}
    \big| \bE_{P_{W,Z_i}} \big[ \textnormal{gen}_i(W,Z_i) \big] \big| \leq \sqrt{2\sigma^2 I(W;Z_i)},
\end{equation*}
where we let $Q_Y = P_Y = P_W$.

Finally, if we note that $\bE_{P_{W,S}}[\ell(W,Z_i)] = \bE_{P_{W,Z_i}}[\ell(W,Z_i)]$ and we use the triangle inequality we get the desired result as follows:
\begin{align*}
\pushQED{\qed} 
    \big| \bE_{P_{W,S}} \big[ \textnormal{gen}(W,S) \big] \big| &\leq 
    \left| \bE_{P_W \times P_Z}[\ell(W,Z)] -  \bE_{P_{W,S}}\left[ \frac{1}{N} \sum_{i=1}^N \ell(W,Z_i) \right] \right| \nonumber \\ 
    &= \frac{1}{N} \left| \sum_{i=1}^N \big( \bE_{P_W \times P_Z}[ \ell(W,Z) ] -  \bE_{P_{W,Z_i}}[ \ell(W,Z_i) ] \big) \right| \nonumber \\ 
    & \leq \frac{1}{N} \sum_{i=1}^N \left| \bE_{P_{W,Z_i}} \big[ \textnormal{gen}_i(W,Z_i) \big] \right| \nonumber \\
    &\leq \frac{1}{N} \sum_{i=1}^N \sqrt{2 \sigma^2 I(W;Z_i)}. \qedhere
\popQED
\end{align*}

\subsection{Proof of Proposition~\ref{prop:negrea_t24_recovery}}

Consider the setting of Lemma~\ref{lemma:abstraction_hellstrom} and, for a fixed $J$, let $X = S_J$, $Y = (W|S_{J^c},R)$, and $f(X,Y) = \textnormal{gen}_J(W,S_J) \triangleq \bE_{P_Z}[\ell(w,Z)] - \frac{1}{M} \sum_{i \in J} \ell(w,Z_i)$.
Then, if (i) $\ell(w,Z)$ is $\sigma$-subgaussian under $P_Z$ for all $w \in \cW$, we have that $\textnormal{gen}_J(W,S_J)$ is zero-mean and $\sigma/\sqrt{M}$-subgaussian under $P_{S_J}$; and if (ii) $\ell(W,Z)$ is $\sigma$-subgaussian under $P_{W|S_{J^c},R} \times P_Z$, then $\textnormal{gen}_J(W,S_J)$ is zero-mean and $\sigma/\sqrt{M}$-subgaussian under $P_{W|S_{J^c},R} \times P_{S_J}$.
Therefore, we have that,
%
%
\begin{equation}
    \big| \bE_{P_{W,S_J|S_{J^c},R }} \big[ \textnormal{gen}_J(W,S_J) \big] \big| \leq  \sqrt{\frac{2\sigma^2}{M} \KL[\big]{P_{W,S_J |S_{J^c},R }}{Q_{W|S_{J^c},R} \times P_{S_J}}}, 
    \label{eq:negrea_reconstruction_abs_value_app_1}
\end{equation}
where $Q_{W|S_{J^c},R}$ is any $(\mathscr{Z}^{\otimes (N-M)} \otimes \mathscr{R} \otimes \mathscr{J})$-measurable distribution over $(\cW,\mathscr{W})$. Then, we may use Jensen's inequality to $|\bE_{P_{W,S,R,J}}[\textnormal{gen}_J(W,S_J)]|$ to obtain that
\begin{equation}
    \big| \bE_{P_{W,S,R,J}} \big[ \textnormal{gen}_J(W,S_J) \big] \big| \leq \bE_{P_{J,S_{J^c},R}} \big[ \big| \bE_{P_{W,S_J|S_{J^c},R }} \big[ \textnormal{gen}_J(W,S_J) \big] \big| \big].
    \label{eq:negrea_reconstruction_abs_value_app_2}
\end{equation}
Finally, we may combine~\eqref{eq:negrea_reconstruction_abs_value_app_1} and~\eqref{eq:negrea_reconstruction_abs_value_app_2} with the fact that $\bE_{P_{W,S,R,J}}[\textnormal{gen}_J(W,S_J)] = \bE_{P_{W,S}}[\textnormal{gen}(W,S)]$ to obtain~\eqref{eq:negrea_t24_recovery}.
\hfill\qedsymbol

\subsection{Proof of Proposition~\ref{prop:icmi}}

Consider the setting of Lemma~\ref{lemma:abstraction_hellstrom} and let $X = U_i$, $Y = (W,\tilde{Z}_i,\tilde{Z}_{i+N})$, and $f(X,Y) = \widehat{\textnormal{gen}}_i(W,\tilde{Z}_i,\tilde{Z}_{i+N}, U_i) \triangleq \ell \big( W,\tilde{Z}_{i + (1-U_i)N} \big) - \ell \big( W,\tilde{Z}_{i+U_iN} \big)$.
Then, since $\ell(W,Z)$ is bounded in $[a,b]$, we also know that $\widehat{\textnormal{gen}}_i(W,\tilde{Z}_i,\tilde{Z}_{i+N}, U_i)$ is bounded in $[a-b,b-a]$, and thus it is $(b-a)$-subgaussian under any source of randomness.
Furthermore, since it has mean zero under $P_{U_i}$, we have that
\begin{equation*}
    \big| \bE_{P_{\smash{W,\tilde{Z}_i,\tilde{Z}_{i+N},U_i}}} \big[ \widehat{\textnormal{gen}}_i(W,\tilde{Z}_i,\tilde{Z}_{i+N}, U_i) \big] \big| \leq \sqrt{2(b-a)^2 I(W;U_i|\tilde{Z}_i,\tilde{Z}_{i+N})},
\end{equation*}
where we let $Q_{\smash{W,\tilde{Z}_i,\tilde{Z}_{i+N}}} = P_{\smash{W,\tilde{Z}_i,\tilde{Z}_{i+N}}}$.
Finally, if we note that $\bE_{P_{\smash{W,\tilde{S},U}}} \big[ \widehat{\textnormal{gen}} (W,\tilde{S},U) \big] = \bE_{P_{W,S}} \big[  \textnormal{gen}(W,S) \big]$, that
\begin{equation*}
    \bE_{P_{\smash{W,\tilde{S},U}}} \big[  \widehat{\textnormal{gen}}(W,\tilde{S},U) \big] 
    = \frac{1}{N} \sum_{i=1}^N \bE_{P_{\smash{W,\tilde{Z}_i,\tilde{Z}_{i+N},U_i}}}  \big[ \widehat{\textnormal{gen}}_i(W,\tilde{Z}_i,\tilde{Z}_{i+N}, U_i) \big],
\end{equation*}
and we use the fact that $|\sum_{i=1}^N x_i| \leq \sum_{i=1}^N |x_i|$, we obtain the result from~\eqref{eq:icmi_result}.
\hfill\qedsymbol

\subsection{Proof of Proposition~\ref{prop:cond_version_negrea_t24}}

Consider the setting of Lemma~\ref{lemma:abstraction_hellstrom} and, for a fixed value of $J$, let $X=U_J$, $Y=(W|\tilde{S},U_{J^c},R)$, and $f(X,Y) = \widehat{\textnormal{gen}}_J(W,U_J,\tilde{S}_J, \tilde{S}_{J+N}) \triangleq \frac{1}{M} \sum_{i \in J} \ell \big( W,\tilde{Z}_{i+(1-U_i)N} \big) - \ell \big( W,\tilde{Z}_{i+U_iN} \big)$.
Then, since $\ell(W,Z)$ is bounded in $[a,b]$, we also know that $\ell \big( W,\tilde{Z}_{i+(1-U_i)N} \big) - \ell \big( W,\tilde{Z}_{i+U_iN} \big)$ is bounded in $[a-b,b-a]$, and thus it is $(b-a)$-subgaussian under $P_{U_J}$.
Hence, since all the terms in the sum of $\widehat{\textnormal{gen}}_J(W,U_J,\tilde{S}_J, \tilde{S}_{J+N})$ are independent of each other under $P_{U_J}$, we know that $\widehat{\textnormal{gen}}_J(W,U_J,\tilde{S}_J, \tilde{S}_{J+N})$ is $(b-a)/\sqrt{M}$-subgaussian under $P_{U_J}$.
Furthermore, since it has zero mean under $P_{U_J}$, we have that
%
%
\begin{equation}
    \big| \bE_{P_{\smash{W,U_J|\tilde{S},U_{J^c},R}}} \big[ \widehat{\textnormal{gen}}_J(W,U_J,\tilde{S}_J, \tilde{S}_{J+N}) \big]  \big| \leq \sqrt{\frac{2(b-a)^2}{M} \KL[\big]{P_{\smash{W,U_J|U_{J^c},\tilde{S},R}}}{Q_{\smash{W|U_{J^c},\tilde{S},R}} \times P_{U_J}}},
    \label{eq:cond_version_negrea_t24_abs_value_app_1}
\end{equation}
where $Q_{\smash{W|U_{J^c},\tilde{S},R}}$ is any $(\mathscr{U}^{\otimes (N - M)} \otimes \mathscr{Z}^{\otimes 2N} \otimes \mathscr{R} \otimes \mathscr{J})$-measurable distribution over $(\cW,\mathscr{W})$.
Then, we may use Jensen's inequality to $|\bE_{P_{W,\tilde{S},U,R,J}}[\widehat{\textnormal{gen}}_J (W, U_J, \tilde{S}_J, \tilde{S}_{J+N})]|$ to obtain that
\begin{equation}
    \big| \bE_{P_{W,\tilde{S},U,R,J}}\big[\widehat{\textnormal{gen}}_J (W,U_J,\tilde{S}_J,\tilde{S}_{J+N})\big] \big| \leq \bE_{P_{\smash{J,\tilde{S},U_{J^{c}},R}}} \big[ \big| \bE_{P_{\smash{W,U_J|\tilde{S},U_{J^c},R}}} \big[ \widehat{\textnormal{gen}}_J(W,U_J,\tilde{S}_J, \tilde{S}_{J+N}) \big]  \big| \big].
    \label{eq:cond_version_negrea_t24_abs_value_app_2}
\end{equation}
Finally, we may combine~\eqref{eq:cond_version_negrea_t24_abs_value_app_1} and~\eqref{eq:cond_version_negrea_t24_abs_value_app_2} with the fact that $\bE_{P_{\smash{W, \tilde{S}, U, R, J}}}[\widehat{\textnormal{gen}}_J(W,U_J,\tilde{S}_J,\tilde{S}_{J+N})] = \bE_{P_{W,S}}[\textnormal{gen}(W,S)]$ to obtain~\eqref{eq:cond_version_negrea_t24}.
\hfill\qedsymbol

\subsection{Proof of Proposition~\ref{prop:mini-batch-ld}}

First of all, we note that 
\[
\KL[\big]{P_{\smash{W|U,\tilde{S},V^T}}}{Q_{\smash{W|U_{J^c},\tilde{S},V^T}}} \leq \sum_{t=1}^T \bE_{P_{\smash{W^{t-1}|U, \tilde{S},V^{t-1}}}} \big[ \KL[\big]{P_{\smash{W_t|W_{t-1},U,\tilde{S},V^t}}}{Q_{\smash{W|W^{t-1}, U_{J^c},\tilde{S},V^t}}},
\]
since we first substitute $W_T$ by $W^T$ using the monotonicity of the relative entropy, and then we employ the chain rule of the relative entropy~\citep[Theorem 2.2]{polyanskiy2014lecture} and the Markov properties of the problem.
Note that we only used the Markov property (with respect to the dependency with previous hypotheses) of the problem for the distribution $P_{\smash{W^T|U,\tilde{S},V^T}}$ and not for $Q_{\smash{W^T|U_{J^c},\tilde{S},V^T}}$, since we will benefit from the hypotheses' trajectory when estimating $U_J$ with $\pi_{J,t}$.
After that, we note that at the iterations $t^*$ where the batch $V_{t^*}$ did not include the sample $J$, the hypothesis $W_{{t^*}}$ is completely determined by $W_{{t^* - 1}}, U_{J^c}, \tilde{S},$ and $V^{t^*}$.
Hence, as in~\citep{bu2020tightening}, we can drop these terms in the summation by considering $Q_{\smash{W|W^{t^*-1}, U_{J^c}, \tilde{S}, V^t}} = P_{\smash{W_t|W_{t^*-1}, U_{J^c}, \tilde{S}, V^t}}$. That is,
\begin{multline*}
 \sum_{t=1}^T \bE_{P_{\smash{W^{t-1}|U, \tilde{S}, V^{t-1}}}} \big[ \KL[\big]{P_{\smash{W_t|W_{t-1}, U, \tilde{S}, V^t}}}{Q_{\smash{W|W^{t-1}, U_{J^c}, \tilde{S}, V^t}}} \big] \\
 = \sum\nolimits_{t \in \mathcal{T}_J(V^T)} \bE_{P_{\smash{W^{t-1}|U, \tilde{S}, V^{t-1}}}} \big[ \KL[\big]{P_{\smash{W_t|W_{t-1}, U, \tilde{S}, V^t}}}{Q_{\smash{W|W^{t-1}, U_{J^c}, \tilde{S}, V^t}}} \big],
\end{multline*}
where $\mathcal{T}_J(V^T)$ is the set of iterations for which sample $J$ was included in the batches from $V^T$.

Then, we can evaluate the terms $\KL[\big]{P_{\smash{W_t|W_{t-1}, U, \tilde{S}, V^t}}}{Q_{\smash{W|W^{t-1}, U_{J^c}, \tilde{S}, V^t}}}$, where similarly to~\citep{haghifam2020sharpened}, we note that $P_{\smash{W_t|W_{t-1}, U, \tilde{S}, V^t}} = \mathcal{N}(\mu_{J,t}, \sigma_t I_d)$, where 
\begin{equation*}
	\mu_{J,t} = \theta_{t-1} - \frac{\eta_t}{|V_t|} \left( (|V_{t}|-1) \nabla_{\theta_{t-1}} L_{S_{V_t \setminus J}}(W_{\theta_{t-1}}) +  (1-U_J) \nabla_{\theta_{t-1}} \ell(W_{\theta_{t-1}},\tilde{Z}_{J}) + U_J  \nabla_{\theta_{t-1}} \ell(W_{\theta_{t-1}},\tilde{Z}_{J+N}) \right).
\end{equation*}
We may also define $Q_{\smash{W|W^{t-1}, U_{J^c}, \tilde{S}, V^t}}$ as a Gaussian $\mathcal{N}(\mu'_{J,t},\sigma_t I_d)$, but since we do not know $U_J$, we consider a weighted average of the gradients $ \nabla_{\theta_{t-1}} \ell(W_{\theta_{t-1}},\tilde{Z}_{J})$ and $\nabla_{\theta_{t-1}} \ell(W_{\theta_{t-1}},\tilde{Z}_{J+N})$ with an estimate $\pi_{J,t}$ of the probability of $U_{J} = 1$ based on $W^{t-1}, \tilde{S}$, $U_{J^c}$, and $V^t$. That is, 
\begin{equation*}
	\mu'_{J,t} = \theta_{t-1} - \frac{\eta_t}{|V_t|} \left( (|V_{t}|-1) \nabla_{\theta_{t-1}} L_{S_{V_t} \setminus J}(W_{\theta_{t-1}}) +  (1-\pi_{J,t}) \nabla_{\theta_{t-1}} \ell(W_{\theta_{t-1}},\tilde{Z}_{J}) + \pi_{J,t} \nabla_{\theta_{t-1}} \ell(W_{\theta_{t-1}},\tilde{Z}_{J+N}) \right).
\end{equation*}

Finally, similarly to~\citep{haghifam2020sharpened}, we can use the analytical expression of the divergence between two Gaussians to obtain that
\begin{equation*}
	\KL[\big]{\mathcal{N}(\mu_{J,t},\sigma_t I_d)}{\mathcal{N}(\mu'_{J,t},\sigma_t I_d)} = (U_J - \pi_{J,t})^2 \frac{\eta_t^2 \|\zeta_{J,t}\|^2}{2 \sigma^2 |V_t|^2},
\end{equation*}
where $\zeta_{J,t} \triangleq \nabla_{\theta_{t-1}} \ell(W_{\theta_{t-1}},\tilde{Z}_{J}) - \nabla_{\theta_{t-1}} \ell(W_{\theta_{t-1}},\tilde{Z}_{J+N})$ is the two-sample incoherence at iteration $t$.
\hfill\qedsymbol

\subsection{Proof of Proposition~\ref{prop:mini-batch-ld-mixture}}

The proof follows that of Proposition~\ref{prop:mini-batch-ld}, except that we construct $Q_{\smash{W|W^{t-1}, U_{J^c}, \tilde{S}, V^t}}$ as a mixture of two Gaussians weighted by an estimate $\pi_{J,t}$ of the probability that $U_J = 1$.
That is, $(1-\pi_{J,t}) \mathcal{N}(\mu^0_{J,t},\sigma_t I_d) + \pi_{J,t} \mathcal{N}(\mu^1_{J,t},\sigma_t I_d)$, where
\begin{align*}
\mu^0_{J,t} &= \theta_{t-1} - \frac{\eta_t}{|V_t|} \left( (|V_{t}|-1) \nabla_{\theta_{t-1}} L_{S_{V_t} \setminus J}(W_{\theta_{t-1}}) +  \nabla_{\theta_{t-1}} \ell(W_{\theta_{t-1}},\tilde{Z}_{J}) \right),  \\ 
\mu^1_{J,t} &= \theta_{t-1} - \frac{\eta_t}{|V_t|} \left( (|V_{t}|-1) \nabla_{\theta_{t-1}} L_{S_{V_t} \setminus J}(W_{\theta_{t-1}}) +  \nabla_{\theta_{t-1}} \ell(W_{\theta_{t-1}},\tilde{Z}_{J+N}) \right).
\end{align*}
Then, we apply~\citep[Lemma 2]{rodriguez2020upper} to bound the relative entropy $\KL[\big]{P_{\smash{W_t|W_{t-1}, U, \tilde{S}, V^t}}}{Q_{\smash{W|W^{t-1}, U_{J^c}, \tilde{S}, V^t}}}$, where we obtain that 
\begin{multline}
	\KL[\big]{\mathcal{N}(\mu_{J,t},\sigma_t I_d)}{(1-\pi_{J,t})  \mathcal{N}(\mu^0_{J,t},\sigma_t I_d) + \pi_{J,t} \mathcal{N}(\mu^1_{J,t},\sigma_t I_d)} \\ 
	\leq - \log\left( (1-\pi_{J,t}) \exp\left(- U_J^2 \frac{\eta_t^2 \|\zeta_{J,t}\|^2}{2\sigma_t^2|V_t|^2}\right) + \pi_{J,t} \exp \left(- (1-U_J)^2 \frac{\eta_t^2 \|\zeta_{J,t}\|^2}{2\sigma_t^2|V_t|^2} \right) \right). \qquad \qquad \qquad
	\label{eq:mini-batch-ld-mixture-intermediate}
\end{multline}

Finally, we can make~\eqref{eq:mini-batch-ld-mixture-intermediate} more compact by noting that the equation reduces to
\[
	- \log\left( (1-\pi_{J,t}) + \pi_{J,t} \exp \left(- \frac{\eta_t^2 \|\zeta_{J,t}\|^2}{2\sigma_t^2|V_t|^2} \right) \right)
\]
when $U_J = 0$, and to
\[
- \log\left( (1-\pi_{J,t}) \exp\left(- \frac{\eta_t^2 \|\zeta_{J,t}\|^2}{2\sigma_t^2|V_t|^2}\right) + \pi_{J,t} \right)
\]
when $U_J = 1$.
Hence, since $|U_J - \pi_{J,t}|$ is equal to $\pi_{J,t}$ when $U_J = 0$, and $(1-\pi_{J,t})$ when $U_J = 1$, we may write the right hand side of~\eqref{eq:mini-batch-ld-mixture-intermediate} as 
\begin{equation*}
-\log \left(|U_J - \pi_{J,t}| \exp \left(- \frac{\eta_t^2 \|\zeta_{J,t} \|^2}{2\sigma_t^2|V_t|^2} \right) + |1 - U_J - \pi_{J,t}| \right),
\end{equation*}
from which we obtain~\eqref{eq:mini-batch-ld-mixture}.
\hfill\qedsymbol

\subsection{{Reduction of \texorpdfstring{\citep[Theorem 3.1]{negrea2019information}}{Negrea et al. [Theorem 3.1]}, under the conditions of Corollary~\ref{cor:mini-batch-ld-lipschitz}}}

If we consider $|J| = 1$ and a bounded loss function $\ell: \cW \times \cZ \rightarrow [a,b]$, we have that~\citep[Theorem 3.1]{negrea2019information} is
\begin{equation*}
    \bE_{P_{W,S}} \big[ \textnormal{gen}(W,S) \big] \leq \frac{(b-a)}{2\sqrt{2}} \,
 \bE_{P'_1} \left[\sqrt{\sum_{\smash{t = 1}}^T \bE_{{P'_2}^t} \bigg[ \frac{\eta_t^2}{\sigma_t^2} \| \xi_t \|^2 \bigg]} \right],
\end{equation*}
where $P'_1 = P_{J,S,V^T}$, ${P'_2}^t = P_{W^{t-1}|S,V^{t-1}}$, and $\xi_t$ is defined as
\[
\xi_t \triangleq \frac{|V_t| - |S_{J^c} \cap S_{V_t}|}{|V_t|} \left( \nabla_{\theta_{t-1}} L_{S_1}(W_{\theta_{t-1}}) - \nabla_{\theta_{t-1}} L_{S_2}(W_{\theta_{t-1}})\right), 
\]
where $S_1 = S_{V_t} \setminus S_{J^c}$ and $S_2 = S_{J^c}$.
Note that if the loss is $L$-Lipschitz and we fix $|V_t| = K$, we have that if $J \in V_t$ then $\| \xi_t \| = \frac{2L}{K}$ and otherwise $\| \xi_t \| = 0$.
Therefore, we can eliminate from the summation all the iterations for which $J$ does not belong to the batch $V_t$, which leaves us with 

\begin{equation*}
 \bE_{P_{W,S}} \big[ \textnormal{gen}(W,S) \big] \leq \frac{L(b-a)}{\sqrt{2}K} \,
 \bE_{P_{\smash{J, V^T}}} \left[\sqrt{\sum_{\smash{t \in \mathcal{T}_{J}(V^T)}}  \frac{\eta_t^2}{\sigma_t^2} } \right] .
\end{equation*}

\subsection{Derivation of the \texorpdfstring{$U_J = 1$}{U=1} estimate}

Similarly to~\citep{haghifam2020sharpened}, we consider the estimation of $U_J = 1$ with the knowledge of $W^{t-1}, U_{J^c}, \tilde{S}$, and $V^T$.
In order to do so, we consider a function $\phi: \bR \rightarrow [0,1]$ of the log-likelihood ratio of the two hypotheses $\mathbb{H}_1 \triangleq  (U_J = 1)$ and $\mathbb{H}_0 \triangleq  (U_J = 0)$.
That is, 
\begin{equation*}
\pi_{J,t} \triangleq  \phi \left( \log \frac{P_{\smash{U_J| W^{t-1}, \tilde{S}, U_{J^c},V^{t-1}}}(1)}{P_{\smash{U_J| W^{t-1}, \tilde{S}, U_{J^c},V^{t-1}}}(0)} \right),
\end{equation*}
where we use $V^{t-1}$ instead of $V^T$ in the probability distribution due to the Markov properties of the problem. 

Then, since the priors of $U_J$ are equal to $1/2$, i.e., $P_{U_J}(0) = P_{U_J}(1) = 1/2$, the hypotheses $W_t$ are completely defined by the parameters $\theta_t \in \bR^d$, and given that the stochasticity considered is Gaussian, we can assume that the probability density functions $f(W^{t-1}|\tilde{S},U_J^c,V^{t-1},U_J=u)$ exist and that 
\begin{equation*}
    \phi \left( \log \frac{P_{\smash{U_J| W^{t-1}, \tilde{S}, U_{J^c},V^{t-1}}}(1)}{P_{\smash{U_J| W^{t-1}, \tilde{S}, U_{J^c},V^{t-1}}}(0)} \right) = \phi \left( \log \frac{f(W^{t-1}|\tilde{S},U_J^c,V^{t-1},U_J=1)}{f(W^{t-1}|\tilde{S},U_J^c,V^{t-1},U_J=0)} \right),
\end{equation*}
where the density functions are
\begin{align*}
     f(W^{t-1}|&\tilde{S},U_J^c,V^{t-1},U_J=u) = \nonumber \\
     &\prod_{t = 1}^T \left(\frac{1}{2\pi \sigma_t}\right)^{\frac{d}{2}} \exp \left( - \frac{1}{2\sigma_t^2} \left \| \theta_t - \theta_{t-1} + \eta_t \frac{|V_t|-1}{|V_t|} \nabla_{\theta_{t-1}} L_{V_t \setminus J}(W_{\theta_{t-1}}) + \frac{\eta_t}{|V_t|} \nabla_{\theta_{t-1}} \ell(W_{\theta_{t-1}},\tilde{Z}_{J+uN}) \right \|^2\right).
\end{align*}

Therefore, the final estimation is
\begin{equation*}
 \pi_{J,t} \triangleq  \phi \left( \sum\nolimits_{t \in \mathcal{T}_J(V^t)} (Y_{J,t,0} - Y_{J,t,1}) \right),
\end{equation*}
where 
\begin{equation*}
    Y_{J,t,u} = \frac{1}{2\sigma_t^2} \left \| \theta_t - \theta_{t-1} + \eta_t \frac{|V_t|-1}{|V_t|} \nabla_{\theta_{t-1}} L_{V_t \setminus J}(W_{\theta_{t-1}}) + \eta_t \frac{1}{|V_t|} \nabla_{\theta_{t-1}} \ell(W_{\theta_{t-1}},\tilde{Z}_{J+uN}) \right \|^2,
\end{equation*}
and where the iterations $t^*$ where $J$ did not belong to $V_t$ are dropped since $Y_{J,t^*,0} = Y_{J,t^*,1}$.


\subsection{Comparison of mutual-information--based bounds on the randomized subsample setting}
\label{app:individual_cmi_tightest}

It is known that in the standard setting, the expected generalization error bounds are ordered, from tightest to loosest, as follows: \citep[Proposition~1]{bu2020tightening}, \citep[Theorem~1]{xu2017information}, and \citep[Theorem~2.4, after further applying Jensen's inequality]{negrea2019information}. This ordering follows from the fact that:
\begin{equation*}
    \sum_{i=1}^N I(W;Z_i) \leq I(W;S) = \sum_{i=1}^N I(W;Z_i|S^{i-1}) \leq \sum_{i=1}^N I(W;Z_i|S^{-i}),
\end{equation*}
where $S^{-i} = S \setminus Z_i$ and $S^{i} = (S_1, \ldots, S_i)$.
The proof for the first inequality is presented in~\citep[Proposition~2]{bu2020tightening} and the second inequality can be derived similarly.

\begin{lemma}
In the standard setting, we have that $I(W;Z_i|S^{i-1}) \leq I(W;Z_i|S^{-i})$ for all $i \in [N]$.
\end{lemma}
\begin{proof}
If we exploit the fact that $Z_i$ are independent we have that:
\begin{align*}
    I(W;Z_i|S^{-i}) &= I(W;Z_i|S^{-i}) + I(Z_i;S^{-i}) \\
    &\stack{a}{=} I((W,S^{-i});Z_i) \\ 
    &\stack{b}{=} I(S^{i-1};Z_i) + I(W;Z_i|S^{i-1}) + I(S_{i+1}^N;Z_i|W,S^{i-1}) \\
    &\stack{c}{\geq} I(W;Z_i|S^{i-1}),
\end{align*}
where for $i < j$ we have that $S_{i}^{j} = (Z_i, \ldots, Z_j)$  and for $i = j$ we have that $S_{i}^{j} = Z_i$. Then, in (a) and (b) we used the chain rule of the mutual information~\citep[Theorem~2.5]{polyanskiy2014lecture} and in (c) the fact that the mutual information is non-negative.
\end{proof}

We can use the same arguments in the bounds for the randomized subsample setting, ordering them as follows: \citep[Theorem~3.4, after further applying Jensen's inequality]{haghifam2020sharpened}, \citep[Theorem~5.1]{steinke2020reasoning}, and \citep[Theorem~3.7, after further applying Jensen's inequality]{haghifam2020sharpened}. This ordering, as before, follows since
\[
\sum_{i=1}^N I(W;U_i|\tilde{S}) \leq I(W;U|\tilde{S}) = \sum_{i=1}^N I(W;U_i|\tilde{S},U^{i-1}) \leq \sum_{i=1}^N I(W;U_i|\tilde{S},U^{-i}),
\]
where $U^{-i} = U \setminus U_i$ and $U^{i} = (U_1, \ldots, U_i)$.
Nonetheless, the bound obtained in Proposition~\ref{prop:icmi} is the tightest, when the bounds are written in their mutual information form after applying Jensen's inequality, as dictated by the following lemma.

\begin{lemma}
In the randomized subsample setting, we have that $I(W;U_i|\tilde{Z}_i,\tilde{Z}_{i+N}) \leq I(W;U_i|\tilde{S})$ for all $i \in [N]$.
\end{lemma}
\begin{proof}
We may start the proof by noting that
\begin{equation}
    I(W;U_i|\tilde{S}) = I(U_i; (W,\tilde{S})) - I(U_i,\tilde{S}) = I(U_i;(W,\tilde{S})),
    \label{eq:single_index_bound_loose}
\end{equation}
where we used the small chain rule of the mutual information~\citep[Theorem~2.5]{polyanskiy2014lecture} and the fact that $U$ and $\tilde{S}$ are independent.
Similarly, we have that
\begin{equation}
    I(W;U_i|\tilde{Z}_{i}, \tilde{Z}_{i+N}) = I(U_i;(W,\tilde{Z}_{i}, \tilde{Z}_{i+N})).
    \label{eq:single_index_bound_tight}
\end{equation}
Then, we may operate with $I(W;U_i|\tilde{S})$ in the form of~\eqref{eq:single_index_bound_loose} to obtain that it is greater or equal than $I(W;U_i|\tilde{Z}_i,\tilde{Z}_{i+N})$ in the form of~\eqref{eq:single_index_bound_tight}, namely
\begin{align*}
    I(U_i;(W,\tilde{S})) &= I(U_i; (W,\tilde{Z}_i,\tilde{Z}_{i+N},\tilde{S}^{-i}) \\
    &= I(U_i;(W,\tilde{Z}_{i}, \tilde{Z}_{i+N})) + I(U_i;\tilde{S}^{-i}|W,\tilde{Z}_{i}, \tilde{Z}_{i+N}) \\
    &\geq I(U_i;(W,\tilde{Z}_{i}, \tilde{Z}_{i+N})),
\end{align*}
where $\tilde{S}^{-i} = \tilde{S} \setminus (\tilde{Z}_i,\tilde{Z}_{i+N})$ and we used again the small chain rule of mutual information and the fact that the mutual information is non-negative. 
\end{proof}

This way, we have that
\[
    \sum_{i=1}^N I(W;U_i|\tilde{Z}_i,\tilde{Z}_{i+N}) \leq \sum_{i=1}^N I(W;U_i|\tilde{S}) \leq I(W;U|\tilde{S}) = \sum_{i=1}^N I(W;U_i|\tilde{S},U^{i-1}) \leq \sum_{i=1}^N I(W;U_i|\tilde{S},U^{-i}).
\]
\label{sec:appendix}

\end{document}